\documentclass[aps,pra,twocolumn]{revtex4-2}

\usepackage{amsthm}
\usepackage{graphicx}
\usepackage{fancyhdr,fancybox}
\usepackage{varioref,float,amssymb,amsmath}
\usepackage{float,amsfonts} 
\usepackage{newlfont}
\usepackage{makeidx}
\usepackage{flafter}
\usepackage{makeidx}
\usepackage[figuresright]{rotating}
\usepackage{txfonts}

\newcommand{\ket}[1]{\left| #1 \right\rangle}
\newcommand{\bra}[1]{\left\langle #1 \right|}
\newcommand{\braket}[2]{\left\langle #1 \big| #2 \right\rangle}                 
\newcommand{\bracket}[3]{\left\langle #1 \big| #2 \big| #3 \right\rangle}       

\newcommand{\e}{\mathrm{e}}
\renewcommand{\exp}[1]{\mathrm{exp}\left( #1 \right)}
\renewcommand{\i}{\mathrm{i}}
\newcommand{\set}[1]{\left\{ #1 \right\}}
\newcommand{\paren}[1]{\left( #1 \right)} 

\newtheorem{lemma}{Lemma}

\usepackage{color}

\begin{document}

\title{Quantum walk-based search algorithms with multiple marked vertices}
\author{G. A. Bezerra, P. H. G. Lug\~ao, and R. Portugal}
\affiliation{%
 National Laboratory of Scientific Computing (LNCC)\\
 Petr\'opolis, RJ, 25651-075, Brazil 
}%

\begin{abstract}
The quantum walk is a powerful tool to develop quantum algorithms, which usually are based on searching for a vertex in a graph with multiple marked vertices, Ambainis's quantum algorithm for solving the element distinctness problem being the most shining example. 
In this work, we address the problem of calculating analytical expressions of the time complexity of finding a marked vertex using quantum walk-based search algorithms with multiple marked vertices on arbitrary graphs, extending previous analytical methods based on Szegedy's quantum walk, which can be applied only to bipartite graphs. Two examples based on the coined quantum walk on two-dimensional lattices and hypercubes show the details of our method.
\end{abstract}
\maketitle

\section{Introduction}

The discrete-time quantum walk is the quantum counterpart of the classical random walk. In the classical case, the state of a random walker is a probability distribution, and its dynamics is described by a stochastic matrix (acting upon the state), which is obtained from the adjacency matrix of the graph. The adjacency matrix ensures that the random walk obeys locality constraints, which means that if the walker is on a vertex $v$ at time $t$, the walker will hop to some vertex in the neighborhood of $v$ at time $t+1$~\cite{MR96}. In the quantum case, the state of a quantum walker is a $L_2$ norm-1 vector in a Hilbert space, and its dynamics is described by unitary matrices, as demanded by the laws of quantum mechanics, but not only that, the unitary matrices must be \textit{local}. The locality is defined by some discrete structure, which characterizes a neighborhood for each allowed location for the quantum walker on that discrete structure~\cite{Por18book}. Most papers in literature employ graphs and the allowed locations are vertices, edges, arcs, or faces, depending on the quantum walk model~\cite{KPSS18}. The model is a recipe that provides local unitary operators, and the product of those operators is the evolution operator of the model. The quantum walk is not only an important tool to build quantum algorithms that outperform their classical counterparts~\cite{Zho21}, but also a versatile toy model useful to simulate and analyze complex physical systems~\cite{MRLA08,Ven12}.

In 2002, Benioff~\cite{Ben02} came up with the idea of quantum robots searching a two-dimensional lattice for a specific site, which inspired many researchers to use quantum walks for searching algorithms. In the same year, Shenvi~\textit{et al.}~described a coined quantum walk-based search algorithm on hypercubes with a quadratic improvement over a classical random walk-based algorithm. Ambainis \textit{et al.}~described a similar search algorithm on two-dimensional lattices~\cite{AKR05}, which was improved by Tulsi in 2008 by adding a qubit to the model~\cite{Tul08}. Tulsi's modification proved useful for other graphs~\cite{Tul12}.

The first quantum walk-based search algorithm (on a specific bipartite graph) with multiple marked vertices was designed by Ambainis~\cite{Amb07a} in 2003, and in this case, the searching has a practical application for solving the element distinctness problem. The algorithm time-complexity is calculated via a reduction method, which converts the marked set into only one marked vertex in a reduced graph.
Ambainis's quantum walk was extended to arbitrary symmetric bipartite graphs with multiple marked vertices by Szegedy~\cite{Sze04a}, who was able to obtain a quadratic improvement for the detection problem, which aims to determine whether there is at least one marked vertex in the graph. The searching problem, which aims to determine where is the location of at least one marked vertex in the graph, cannot be solved with an quadratic speedup on arbitrary graphs, but can be solved with a quadratic speedup on \textit{bipartite graphs} by using a combination of Szegedy's quantum walk, the phase-estimation algorithm, and interpolated quantum walks~\cite{KMOR16,AGJK20}. In technical terms, Szegedy showed that the quantum hitting time of a quantum walk on a bipartite graph is~$\sqrt{h}$, where $h$ is the hitting time of a classical Markov chain on the underlying graph. In Szegedy's model, the underlying graph can be any graph, but the graph on which the quantum walk takes place must be bipartite.

The coined model~\cite{AAKV01} is a recipe to define quantum walks on graphs by extending the position space with an internal coin space. The dimension of the coin-position Hilbert space is strictly larger than the number of vertices. This coin extension can be understood via graph theory as a modification of the graph itself by inserting for each vertex a clique whose size is equal to the degree of the vertex~\cite{Por16}, so that the number of vertices of the extended clique-inserted graph~\cite{ZCC09} is equal to the dimension of the original coin-position Hilbert space. By inserting cliques, the extended graph is nonbipartite if there is at least one degree-3 vertex in the original graph, for instance, the $n$-dimensional hypercube with $n>2$. This means that the results about finding at least one marked vertex on bipartite graphs help little for the coined model.

Many papers have addressed the searching problem with multiple marked vertices using the coined model~\cite{WS17,AP18,LS20,GNBK021}, and most of them show negative results by displaying ``exceptional configurations'' of marked vertices that hinder the coined quantum walk-based searching algorithms. Ref.~\cite{AP18} shows how to eliminate exceptional configuration by using Tulsi's modification on regular graphs~\cite{Tul08}. It is not known whether those unwelcome configurations are really exceptional or in fact common in the standard coined model. To answer this kind of question we need analytical methods to address this problem. In this work, we describe an analytical framework for determining the time complexity of discrete-time quantum walk-based searching algorithms on arbitrary graphs with multiple marked vertices, which can be applied not only to the coined model but also to any discrete-time quantum walk. The standard dynamics of searching algorithms is based on a modification the underlying evolution operator $U$ of the quantum walk by multiplying $U$ by an unitary operator $R$, usually called \textit{oracle}, that knows the locations of the marked vertices, so that the new evolution operator $U'$ is $U\cdot R$.  Our framework uses two eigenvectors of $U'$ whose eigenvalues are closest to 1 but different from 1, extending a similar method that was successfully used to analyze quantum walk-based search algorithms on many graphs with \textit{only} one marked vertex~\cite{SKW03,AKR05,Por18book}.

We provide two examples of our method. We calculate analytically the time complexity of quantum walk-based search algorithms on two-dimensional lattices and hypercubes with two marked vertices each, using the coined model with the Grover coin. We show that the asymptotic optimal running time and the success probability depend on the relative distance of the marked vertices for the two-dimensional lattice. The speed of the algorithm is slower when the marked vertices are neighbors. For the $n$-dimensional hypercube, the asymptotic optimal running time is $\pi\sqrt{2^n}/4$ and the success probability is $1/2$, regardless the locations of the marked vertices. The calculations can be extended to more marked vertices, and numerical experiments show that the quantum-walk based search algorithm on the hypercube is similar to Grover's algorithm~\cite{Gro97} in the sense that the optimal running time $t_m$ for $m$ marked vertices is $t_m=t_1/\sqrt{m}$ with success probability 1/2 (no dependence on $m$). 

The structure of the paper is as follows. Sec.~\ref{ABA} describes the method to determine the time-complexity of quantum walk-based search algorithms on arbitrary graphs with multiple marked vertices, and gives all the details when there are two marked vertices.
Sec.~\ref{sec:aba_MalhaFinita} applies the method to two-dimensional lattices with two marked vertices.
Sec.~\ref{sec:Hypercube} applies the method to hypercubes with two marked vertices. 
Sec.~\ref{sec:hypercubeM3} shows how numerical methods can be improved.
Sec.~\ref{sec:conc} presents our conclusions.

\section{Time complexity of search algorithms with multiple marked vertices} \label{ABA}

Consider a graph $\Gamma$, where $V(\Gamma)$ is the set of vertices and $\left|V(\Gamma)\right|=N$. Let ${\mathcal{H}}^N$ be the $N$-dimensional Hilbert space associated with the graph, that is, the computational basis of ${\mathcal{H}}^N$ is $\{\ket{v}\,:\, v\in V(\Gamma)\}$. Although the dimension of Hilbert space is equal to the number of vertices, the results of this section can be applied to the coined model, as we show in concrete examples.

Let $M$ be the set of marked vertices. Then, the unitary operator that marks a vertex $v$ is
\begin{equation}\label{eq:sqw_R_multimarked}
R\,=\,I-2\sum_{v\in M}\ket{v}\bra{v}.
\end{equation} 
The evolution operator $U'$ of a \textit{quantum walk-based search algorithm}  is
\begin{equation}
U'\,=\,UR.
\end{equation}
The walker starts at an initial state $\ket{\psi(0)}$ and evolves driven by $U'$, that is, the walker's state after $t$ steps is $\ket{\psi(t)}=(U')^t\ket{\psi(0)}$.

The probability of finding a marked vertex $m\in M$ after $t$ steps is

\begin{equation}\label{eq:sqwpoft1}
	p(t)=\sum_{m\in M}\left|\bracket{m}{\left(U'\right)^t}{\psi(0)}\right|^2.
\end{equation}
The goal now is to determine the optimal number of steps $t_\textrm{opt}$, which is the one that maximizes $p(t)$. The running time is $t_\textrm{opt}$ and the success probability is $p(t_\textrm{opt})$. 

Let $\ket{\lambda}$ and $\ket{\lambda'}$ be the eigenvectors of $U'$ that have the eigenvalues $\textrm{e}^{\textrm{i}\lambda}$ and  $\textrm{e}^{\textrm{i}\lambda'}$ as close as possible to 1, but different from 1. The eigenspace spanned by the other eigenvectors will be disregarded, which cause some supposedly small error. Then
\begin{align}\label{eq:sqw_p(t)}
p(t) =\sum_{m\in M}&\Big|\textrm{e}^{\textrm{i}\lambda t}\,\braket{m}{\lambda}\braket{\lambda}{\psi(0)}+ \nonumber\\ &\textrm{e}^{\textrm{i} \lambda' t}\,\braket{m}{\lambda'}\braket{\lambda'}{\psi(0)} +\epsilon_m \Big|^2,
\end{align}
where $\epsilon_m=\bracket{m}{U_{\textrm{tiny}}^t}{\psi(0)}$, and $U_{\textrm{tiny}}$ acts non-trivially only on the subspace orthogonal to the plane spanned by $\left\{\ket{\lambda},\ket{\lambda'}\right\}$. Our approach can be applied when $|\epsilon_m|$ is much smaller than the absolute value of the remaining terms in the asymptotic limit (large $N$). We disregard $\epsilon_m$ for now and show applications for which $\lim_{N\rightarrow\infty}|\epsilon_m|=0$.

    Let us start by calculating ${\lambda}$ and $\lambda'$. Suppose that $\{\ket{\psi_k}\}$ is an orthonormal eigenbasis of $U$ and $\exp{\textrm{i}\phi_k}$ the corresponding eigenvalues.  Then, 
\begin{equation}\label{bra00lambda}
 \braket{m}{\lambda}\,=\,\sum_{k} \braket{m}{\psi_{k}}\braket{\psi_{k}}{\lambda},
\end{equation}
where $m\in M$. Using $\bracket{\psi_{k}}{U'}{\lambda}=\bracket{\psi_{k}}{U R}{\lambda}$ and supposing that $\lambda\neq \phi_{k}$, $\forall k$, we obtain 
\begin{equation}\label{psi_lambda}
 \braket{\psi_{k}}{\lambda}\,=\,(1+\text{i}b^\lambda_k){\sum_{m\in M}\braket{\psi_{k}}{m}\braket{m}{\lambda}},
\end{equation}
where 
\begin{equation}
b^\lambda_k\,=\,\frac{\sin(\lambda-\phi_k)}{1-\cos(\lambda-\phi_k)}.
\end{equation}
Replacing Eq.~(\ref{psi_lambda}) into Eq.~(\ref{bra00lambda}), we learn that the $|M|$-vector with entries $\braket{m}{\lambda}$ is a $0$-eigenvector of the $|M|$-dimensional Hermitian matrix $\Lambda^\lambda$, where 
\begin{equation}\label{eq:Lambda}
\Lambda^\lambda_{m m'}\,=\, \sum_k b^\lambda_k\,{\braket{m}{\psi_k}}\braket{\psi_{k}}{m'}.
\end{equation}
Then,
\begin{equation}\label{det_C}
\det\left(\Lambda^\lambda\right)\,=\,0.
\end{equation}
We use this equation to find $\lambda$ and the analog equation $\det(\Lambda^{\lambda'})=0$ to find $\lambda'$. Terms $\braket{m}{\lambda}$ and $\braket{m}{\lambda'}$ are calculated up to their norms using the fact that they are $0$-eigenvectors of $\Lambda^\lambda$ and $\Lambda^{\lambda'}$, respectively. The norms of these $0$-eigenvectors are calculated using Eq.~(\ref{psi_lambda}) and the constraint $\sum_{k}  \left|\braket{\psi_{k}}{\lambda}\right|^2=1$. The last missing terms, $\braket{\lambda}{\psi(0)}$ and $\braket{\lambda'}{\psi(0)}$, can be calculated using Eq.~(\ref{psi_lambda}) by assuming that $\ket{\psi(0)}$ is a uniform superposition of the $(+1)$-eigenvectors of $U$ that have nonzero overlap with the target states. 

There are cases so that $\lambda=-\lambda'$ and $\braket{m}{\lambda}\braket{\lambda}{\psi(0)}$~= $-\braket{m}{\lambda'}\braket{\lambda'}{\psi(0)}$ forall $m$, and Eq.~(\ref{eq:sqw_p(t)}) simplifies in the asymptotic limit to
\begin{align}\label{eq:p(t)}
    p(t) = p_\text{succ} \sin^2\lambda t,
\end{align}
where
\begin{equation}\label{eq:p_succ}
    p_\text{succ}\,=\,4 \sum_{m \in M} \left| \braket{m}{\lambda}\right|^2 \left|\braket{\lambda}{\psi(0)}\right|^2.
\end{equation}
In those cases, we know that the optimal running time is $t_\text{opt}=\pi/(2\lambda)$ and the success probability is $p_\text{succ}$. The time complexity is determined by the asymptotic behavior of  $t_\text{run}=t_\text{opt}/\sqrt{p_\text{succ}}$ as a function of the number of vertices because, in the quantum case after using the amplitude amplification method~\cite{BHMT02}, $t_\text{run}$ is the total running time with success probability $O(1)$.

The results described above can be used not only for analytical calculations but also to speedup numerical methods. We can use Eq.~(\ref{det_C}) to find $\lambda$ numerically for an specific configuration of marked vertices, and then by plotting $t_\text{opt}=\pi/(2\lambda)$ as a function of $N$, we can estimate the complexity of the running time. On the other hand, Eq.~(\ref{eq:p_succ}) can be used to generate a plot of the success probability as a function of $N$. The combination of those plots are used to determine the time complexity of the search algorithm. An example of this numerical approach is shown in Sec.~\ref{sec:hypercubeM3}.

\subsection*{Case $|M|=2$}

Suppose that $M=\{m_{0},m_{1}\}$. Eq.~(\ref{det_C}) implies that 
\begin{equation}\label{ckkpsum}
\sum_{k k'}b^\lambda_k b^\lambda_{k'}c_{k k'}=0,
\end{equation}
where
\begin{align}\label{ckkp}
c_{k k'}\,=\,\psi_k(m_{0})\psi_{k'}(m_{1})\Big(&\psi^*_k(m_{0})\psi^*_{k'}(m_{1})-\nonumber\\ &\psi^*_k(m_{1})\psi^*_{k'}(m_{0})\Big).
\end{align}

Suppose that $\lambda\ll \phi_\textrm{min}$ when $N\gg 1$, where $\phi_\textrm{min}$ is the smallest positive value of $\phi_{k}$. We will check the validity of those assumptions in specific applications. Let us split the sum~(\ref{ckkpsum}) into four parts
\begin{equation}\label{exactsum_kl_2}
 \sum_{k k'}=\sum_{\substack{\phi_k=0\\ \phi_{k'}=0}}+\sum_{\substack{\phi_k=0\\ \phi_{k'}\neq 0}}+\sum_{\substack{\phi_k\neq 0\\ \phi_{k'}=0}}+\sum_{\substack{\phi_k\neq 0\\ \phi_{k'}\neq 0}},
\end{equation} 
corresponding to the sum of terms such that $\phi_k=0$ or $\phi_k\neq 0$. Using 
\begin{equation}\label{eq:sinlam}
 \frac{\sin {\lambda}}{1-\cos \lambda } = \frac{2}{\lambda}+O(\lambda)
\end{equation}
and, if $\phi_k\neq 0$,
\begin{equation}\label{eq:sinlamphi}
 \frac{\sin({\lambda-\phi_{k}})}{1-\cos(\lambda-\phi_{k})} = a_k{\sin{\phi_{k}}}+a_k{\lambda}+O(\lambda^2),
\end{equation}
where
\begin{equation}\label{eq:a_k}
a_k \,=\, \frac{1}{\cos\phi_k-1},
\end{equation}
we obtain 
\begin{equation}\label{eq:lambdaequation}
	\frac{A}{\lambda^2}+\frac{B}{\lambda}+C+D\lambda+E\lambda^2 = O(\lambda^3),
\end{equation}
where
\begin{eqnarray}
	A &=&  4\sum_{\substack{\phi_k=0\\ \phi_{k'}=0}} c_{k k'}  ,\label{eq:sqwA} \\
	B &=& 2\sum_{\substack{\phi_k=0\\ \phi_{k'}\neq 0}}a_{k'} c_{k k'}\sin\phi_{k'} +2\sum_{\substack{\phi_k\neq 0\\ \phi_{k'}=0}}a_{k} c_{k k'}\sin\phi_k,\hspace{15pt}\label{eq:sqwB} \\
	C &=& 2\sum_{\substack{\phi_k=0\\ \phi_{k'}\neq 0}}a_{k'} c_{k k'} +2\sum_{\substack{\phi_k\neq 0\\ \phi_{k'}=0}}a_{k} c_{k k'} 	+\nonumber\\ 
	&&\hspace{20pt}\sum_{\substack{\phi_k\neq 0\\ \phi_{k'}\neq 0}}\,a_{k}\,a_{k'} c_{k k'}\sin\phi_k\sin\phi_{k'},\label{eq:sqwC} \\
	D &=& \sum_{\substack{\phi_k\neq 0\\ \phi_{k'}\neq 0}}a_{k}a_{k'} c_{k k'}\left(\sin\phi_k+ \sin\phi_{k'}\right),\label{eq:sqwD} \\
	E &=& \sum_{\substack{\phi_k\neq 0\\ \phi_{k'}\neq 0}}a_k a_{k'}c_{kk'}.\label{eq:sqwE}
\end{eqnarray}

Using that $\braket{m}{\lambda}$ is a $0$-eigenvector of $\Lambda^\lambda$ and re-scaling $\ket{\lambda}$ by a global phase, we obtain 
\begin{eqnarray}
\braket{m_0}{\lambda}&=& \alpha \Lambda^\lambda_{m_0 m_1} \\
\braket{m_1}{\lambda}&=& -\alpha \Lambda^\lambda_{m_0 m_0}, 
\end{eqnarray}
where $\alpha$ is a positive constant, which can be calculated using $1=\sum_{k}  \left|\braket{\psi_{k}}{\lambda}\right|^2$ and Eq.~(\ref{psi_lambda}):
\begin{align}\label{eq:alpha}
\frac{1}{\alpha^2}\,=\,(&1+\Lambda'_{m_{1}m_{1}})\left(\Lambda^\lambda_{m_0 m_0}\right)^2
-2\Lambda^\lambda_{m_0 m_0}\Re\left\{\Lambda'_{m_{0}m_{1}}\Lambda^\lambda_{m_0 m_1}
\right\} \nonumber \\
&+(1+\Lambda'_{m_{0}m_{0}})\left|\Lambda^\lambda_{m_0 m_1}\right|^2,
\end{align}
where $\Re$ the real part operator,
\begin{equation}\label{eq:Llmmp}
\Lambda^\lambda_{mm'}\,=\, \sum_k b^\lambda_k\,\psi_k(m)\psi^*_k(m'),
\end{equation}
and
\begin{equation}\label{eq:Lpmmp}
\Lambda'_{mm'}\,=\, \sum_k \left(b^\lambda_k\right)^2\psi_k(m)\psi^*_k(m').
\end{equation}
Note that because we have already calculated $\lambda$, we can now calculate explicitly $\Lambda^\lambda_{mm'}$ and $\Lambda'_{mm'}$ if we know the spectral decomposition of $U$.

\section{Searching two-dimensional lattices with 2 marked vertices}\label{sec:aba_MalhaFinita}

As an application, we calculate the time complexity of a quantum walk-based search algorithm on a $\sqrt{N}\times\sqrt{N}$ square lattice assuming that the lattice has cyclic boundary conditions with exactly 2 marked vertices. The evolution operator of a coined quantum walk with no marked vertex is
\begin{equation}\label{aba_U_grid}
	U\,=\,S\,(G\otimes I),
\end{equation}
where $G$ is the Grover coin and $S$ is the flip-flop shift operator given by
\begin{equation}
    S\ket{i,j}\ket{x,y}=\ket{1-i,1-j}\ket{x+(-1)^i,y+(-1)^j},
\end{equation}
where $i$ and $j$ are bits and the arithmetic in the second register is performed modulo $\sqrt{N}$.
To search the lattice, the modified evolution operator is $U'=UR'$, where 
\begin{equation}
	R'=I-2\ket{0'}\bra{0'}-2\ket{1'}\bra{1'}
\end{equation}
and 
\begin{eqnarray}
	\ket{0'}&=&\ket{\text{d}_c}\ket{0,0},\\
	\ket{1'}&=&\ket{\text{d}_c}\ket{x_0,y_0},
\end{eqnarray}
where one marked vertex is chosen at $m_0=(0,0)$ without loss of generality, and the second marked vertex is chosen at position $m_1=(x_0,y_0)$ so that $m_0\neq m_1$. Vector $\ket{\text{d}_c}$ is the normalized uniform superposition of the computational basis of the coined space. Note that here the Hilbert space is larger because it has been augmented by the coin space, and then all formulas of the previous section must be extended by substituting $\ket{\text{d}_c}\ket{m}$ for $\ket{m}$. The initial state $\ket{\psi(0)}$ is the uniform superposition of all states of the computational basis, that is,
\begin{equation}\label{eq:psi0_latt}
	\ket{\psi(0)}=\ket{\text{d}_c}\ket{\text{d}_p},
\end{equation}
where $\ket{\text{d}_p}$ is the normalized uniform superposition of computational basis of the position space. State $\ket{\psi(0)}$ can be generated in $O\big(\sqrt{N}\big)$ steps.

The eigenvectors of $U$ that have nonzero overlap with the marked vertices are $\ket{\psi(0)}$ with eigenvalue 1 and $\ket{\nu_{k\ell}^{\pm\theta}}\ket{\tilde{k},\tilde{\ell}}$ for $(k,\ell)\neq(0,0)$ with eigenvalues $\e^{\pm i\theta_{kl}}$, where~\cite{MPA10,Por18book}
\begin{equation}\label{pq_kappa_x_y}
    \big|{\tilde{k},\tilde{\ell}}\big\rangle=\frac{1}{\sqrt N}\sum_{x,y=0}^{\sqrt{N}-1} \omega^{x k + y \ell} \ket{x,y},
\end{equation}
and
\begin{equation}\label{pq_nu_theta}
    \ket{\nu^{\pm\theta}_{k\ell}}=
\frac{\pm\textrm{i}}{2\sqrt 2 \sin \theta_{k\ell}}
 \left[ \begin {array}{c} {{\textrm e}^{\mp\textrm{i}\theta_{k\ell}}}-{\omega}^{k}
\\\noalign{\medskip}{{\textrm e}^{\mp\textrm{i}\theta_{k\ell}}}-{\omega}^{-k}
\\\noalign{\medskip}{{\textrm e}^{\mp\textrm{i}\theta_{k\ell}}}-{\omega}^{\ell}
\\\noalign{\medskip}{{\textrm e}^{\mp\textrm{i}\theta_{k\ell}}}-{\omega}^{-\ell}
\end {array} \right],
\end{equation}
and
\begin{equation}\label{aba_cos_theta}
    \cos \theta_{k\ell}=\frac{1}{2} \left( \cos {\frac {2 \pi{k}}{\sqrt {N}}} +\cos {\frac {2 \pi{\ell}}{\sqrt{N}}}\right),
\end{equation}
where $\omega=\exp{2\pi\i/\sqrt{N}}.$

Using this list of eigenvectors and $\braket{\text{d}_c}{\nu^{\pm\theta}_{k\ell}}=1/\sqrt{2}$, Eq.~(\ref{ckkp}) \big[with the modification $\psi(m)\rightarrow \left(\bra{\text{d}_c}\bra{m}\right)\ket{\psi}$ and $k\rightarrow j,k,\ell$, where $j$ represents the coin value\big]  reduces to
\begin{equation*}
    c_{k\ell;k'\ell'}=\frac{1}{\epsilon_{k\ell}\epsilon_{k'\ell'}N^2}\left(1-\omega^{(k'-k)x_0+(\ell'-\ell)y_0}\right),
\end{equation*}
where $\epsilon_{k\ell}=1$ if $k=\ell=0$ and $\epsilon_{k\ell}=2$ otherwise. Index $j$ runs from 0 to 3 in the coin space, and can be readily simplified. Then,
we obtain $A=B=D=0$ and
\begin{align}
	C &=  -\frac{4}{N^2} \sum_{\substack{k,\ell=0\\ (k,\ell)\neq (0,0)}}^{\sqrt{N}-1} \frac{1 - \cos {\frac{2\pi({k x_0 + \ell y_0})}{\sqrt{N}} }}{1-\cos \theta_{k\ell}}\label{eq:sqwC_latt2} \\
	E &=  \frac{1}{N^2}\sum_{\substack{k,\ell=0\\ (k,\ell)\neq (0,0)}}^{\sqrt{N}-1}\sum_{\substack{k',\ell'=0\\ (k',\ell')\neq (0,0)}}^{\sqrt{N}-1}
	\frac{1 - \cos\frac{2\pi(k'-k)x_0 + 2\pi(\ell'-\ell)y_0}{\sqrt{N}}}{(1-\cos \theta_{k\ell} )(1-\cos \theta_{k'\ell'})}.\label{eq:sqwE_latt2} 
\end{align}
Note that the imaginary part of $E$ is zero because of symmetry properties. The expression for $\lambda$ reduces to
\begin{equation}\label{eq:lambda}
    \lambda\,=\,\sqrt{\frac{-C}{E}}.
\end{equation}
In order to proceed with the calculations, we introduce the constant $c\in\mathbb{R}$ such that
\begin{equation}\label{eq:S_1}
    \sum_{\substack{k,\ell=0\\ (k,\ell)\neq (0,0)}}^{\sqrt{N}-1} \frac{1}{1-\cos \theta_{k\ell}} = cN\ln{N} + O(N),
\end{equation}
where c is bounded by $2/\pi^2\leq c \leq 1$. Numerical calculations show that $c\approx0.32$.
The asymptotic behavior of $\lambda$ using two representative pairs of marked elements, namely $\big((0,0),(1,0)\big)$ and $\big((0,0),(\sqrt{N}/2,\sqrt{N}/2)\big)$, is obtained in Appendix~\ref{appendA} and is given by
\begin{align}
    \lambda &=& \begin{cases}
        \frac{\sqrt{2}}{\sqrt{c}\sqrt{N\ln{N}}}, &(x_0,y_0)=(1,0),\\
        \frac{2}{\sqrt{c}\sqrt{N\ln{N}}}, &(x_0,y_0)=\left({\sqrt{N}}/{2},{\sqrt{N}}/{2}\right).
    \end{cases}
\end{align}
Note that $\lambda\ll \theta_{k\ell}$ when $N\gg 1$ because $\theta_{k\ell}=O(1/\sqrt{N})$.

\subsection*{Case $(x_0,y_0)=(1,0)$}

Using Eqs.~(\ref{eq:sinlam}), (\ref{eq:sinlamphi}), (\ref{eq:Llmmp}), and lower order terms in the asymptotic expansion of $\lambda$, we obtain
\begin{align}
    \Lambda^\lambda_{m_0m_0}= - \Lambda^\lambda_{m_0m_1}= &-\frac{1}{\sqrt{2c}\sqrt{N\ln{N}}} +\nonumber\\
    & O\left(\frac{1}{\sqrt{N}\ln^{\frac{3}{2}}{N}}\right).
\end{align}
Using Eq.~(\ref{eq:Lpmmp}), we obtain
\begin{align}
     \Lambda'_{m_0m_0}=\Lambda'_{m_0m_1}=\Lambda'_{m_1m_1}=4c\ln{N} + O\left(1\right).
\end{align}
Replacing those results in Eq.~(\ref{eq:alpha}), we have
\begin{equation}
    \alpha = \frac{\sqrt{N}}{2\sqrt{2}}
\end{equation}
and
\begin{equation}
    \braket{0'}{\lambda}=\braket{1'}{\lambda}=\frac{1}{4\sqrt{c}\sqrt{\ln{N}}}.
\end{equation}
Taking the complex conjugate of Eq.~\eqref{psi_lambda} and replacing the above results, we obtain
\begin{equation}
    \braket{\lambda}{\psi(0)}= -\frac{\i}{\sqrt{2}} + O\left(\frac{1}{\sqrt{N\ln{N}}}\right).
\end{equation}
The equation above can be used to show that the terms $\epsilon_m$ in Eq.~(\ref{eq:sqw_p(t)}) tend to zero when $N\rightarrow\infty$.
Using Eqs.~\eqref{eq:p(t)} and~(\ref{eq:p_succ}), we have asymptotically 
\begin{equation}
    p(t)= \frac{1}{4c\ln{N}}\sin^2{\left(\frac{\sqrt{2}\,\,t}{\sqrt{c}\sqrt{N\ln{N}}}\right)}.
\end{equation}
Taking the running time as the optimal $t$, we get
\begin{equation}
    t_\text{opt} = \frac{\pi\sqrt{c}\sqrt{N\ln{N}}}{2\sqrt{2}},
\end{equation}
and
\begin{equation}
    p_\text{succ} = \frac{1}{4c\ln{N}}.
\end{equation}

\subsection*{Case $(x_0,y_0)=\left({\sqrt{N}}/{2},{\sqrt{N}}/{2}\right)$}

In this subsection, we assume that $\sqrt{N}$ is even. Using Eqs.~(\ref{eq:sinlam}), (\ref{eq:sinlamphi}), and (\ref{eq:Llmmp}), in the asymptotic limit we obtain
\begin{equation}
    \Lambda^\lambda_{m_0m_0}=-\Lambda^\lambda_{m_0m_0} =\sqrt{c}\frac{\sqrt{\ln{N}}}{\sqrt{N}} + O\left(\frac{1}{\sqrt{N\ln{N}}}\right).
\end{equation}
Using Eq.~(\ref{eq:Lpmmp}), we obtain
\begin{equation}
     \Lambda'_{m_0m_0}=\Lambda'_{m_1m_1}=3c\ln{N} + O\left(\frac{1}{\ln{N}}\right)
\end{equation}
\begin{equation}
     \Lambda'_{m_0m_1}=c\ln{N} + O\left(1\right)
\end{equation}
Replacing those results in Eq.~(\ref{eq:alpha}), we have
\begin{equation}
    \alpha = \frac{\sqrt{N}}{2c\sqrt{2}\ln{N}}
\end{equation}
and
\begin{equation}
    \braket{0'}{\lambda}=\braket{1'}{\lambda}=\frac{1}{2\sqrt{2c}\sqrt{\ln{N}}}.
\end{equation}
Taking the complex conjugate of Eq.~\eqref{psi_lambda} and replacing the above results, we obtain
\begin{equation}
    \braket{\lambda}{\psi(0)}= -\frac{\i}{\sqrt{2}} + O\left(\frac{1}{\sqrt{N\ln{N}}}\right).
\end{equation}
Using Eqs.~\eqref{eq:p(t)} and~(\ref{eq:p_succ}), we have asymptotically 
\begin{equation}
    p(t)= \frac{1}{2c\ln{N}}\sin^2{\left(\frac{2\,\,t}{\sqrt{c}\sqrt{N\ln{N}}}\right)}.
\end{equation}
Taking the running time as the optimal $t$, we get
\begin{equation}
    t_\text{opt} = \frac{\pi\sqrt{c}\sqrt{N\ln{N}}}{4},
\end{equation}
and
\begin{equation}
    p_\text{succ} = \frac{1}{2c\ln{N}}.
\end{equation}

\section{Searching hypercubes with 2 marked vertices}\label{sec:Hypercube}

As a second application, we calculate the time complexity of a quantum walk-based search algorithm on a hypercube with exactly 2 marked vertices. A hypercube has $N=2^n$ vertices whose labels are binary vectors $\vec v$. The decimal representation of $\vec v$ is in the range $0\le \vec v< N$. The evolution operator of a coined quantum walk with no marked vertex is
\begin{equation}\label{aba_U_hypercube}
	U\,=\,S\,(G\otimes I_N),
\end{equation}
where $G\in {\mathcal{H}}^n$ is the Grover coin and $S\in  {\mathcal{H}}^n\otimes  {\mathcal{H}}^N$ is the flip-flop shift operator given by
\begin{equation}
    S\ket{a}\ket{\vec v}=\ket{a}\ket{\vec v\oplus {\vec e}_a},
\end{equation}
where $1\le a\le n$ and ${\vec e}_a$ is the binary $n$-tuple with all entries zero except the $a$-th entry, whose value is 1.
To search the $n$-dimensional hypercube, the modified evolution operator is $U'=UR'$, where 
\begin{equation}
	R'=I-2\ket{0'}\bra{0'}-2\ket{1'}\bra{1'}
\end{equation}
and 
\begin{eqnarray}
	\ket{0'}&=&\ket{\text{d}_c}\ket{\vec 0},\\
	\ket{1'}&=&\ket{\text{d}_c}\ket{\vec v_0},
\end{eqnarray}
where one marked vertex is chosen at $(0,...,0)$ without loss of generality, and the second marked vertex is chosen at an arbitrary position $\vec v_0$ so that $\vec v_0\neq (0,...,0)$. Vector $\ket{\text{d}_c}$ is the normalized uniform superposition of the computational basis of the coined space. Note that the Hilbert space has been augmented by the coin space, whose basis is $\{\ket{1},...,\ket{n}\}$. The initial state $\ket{\psi(0)}$ is the uniform superposition of all states of the computational basis, that is,
\begin{equation}\label{eq:psi0_hypercube}
	\ket{\psi(0)}=\ket{\text{d}_c}\otimes\ket{\text{d}_p}=\frac{1}{\sqrt{n}}\sum_{a=1}^n\ket{a}\otimes\frac{1}{\sqrt{N}}\sum_{\vec v=0}^{N-1}\ket{\vec v}.
\end{equation}
State $\ket{\psi(0)}$ can be generated in $O\big(\sqrt{N}\big)$ steps using local operators.

\begin{table*}[ht!]
\centering
\begin{tabular}{ccccc}\\
\hline
Hamming wgt. & index~$a$ & eigenval. & $\ket{\alpha^{\vec k}_a}$ & multiplicity \\
\hline
{$k=0$}
         & $a = 1$ & $1$ & $\sum_{b=1}^{n} \ket{b}/\sqrt{n}$ & $1$ \\  
         & $a\in[2,n]$ & $-1$ & $(\ket{1}-\ket{a})/\sqrt{2}$ & $n-1$ \\
\hline
{ $1 \le k \leq n - 1$}
    & $a = q(0)$ & $\e^{\i \omega_k}$ & $\ket{\alpha_+^{\vec k}}$ & $1$\\
    & $a = q(1)$ & $\e^{-\i \omega_k}$ & $\ket{\alpha_+^{\vec k}}^*$ & $1$\\
    & $\set{ a\ |\ k_a = 1} \setminus \set{q(1)}$ & $1$ & $\left( \ket{q(1)}-\ket{a} \right)/\sqrt{2}$ & $k-1$\\
    & $\set{ a\ |\ k_a = 0} \setminus \set{q(0)}$ & $-1$ & $\left( \ket{q(0)}-\ket{a} \right)/\sqrt{2}$ & $n-k-1$\\
\hline
$k = n$
        & $a = 1$ & $-1$ & $\sum_{b=1}^n \ket{b}/\sqrt{n}$ & $1$\\ 
        & $a\in[2,n]$ & $1$ & $(\ket{1}-\ket{a})/\sqrt{2}$ & $n-1$ \\
\hline
\end{tabular}
\caption{Eigenvalues $\exp{\pm\i \omega_k}$ and eigenvectors $\ket{\alpha_a^{\vec k}} \otimes \ket{\beta_{\vec k}}$ of $U$, where $k$ is the Hamming weight of $\vec k$, $\cos \omega_k=1 - 2k/n$, $q(x)=\min\{a\ |\ k_a=x\}$, and $\ket{\alpha_+^{\vec k}}$ is given by Eq.~(\ref{pq_alpha_k+-}). }
        \label{pq_tab_eigenvalvec}
\end{table*}

The results of Sec.~\ref{ABA} can be readily employed as soon as we calculate $A$ to $E$ given by Eqs.~(\ref{eq:sqwA}) to~(\ref{eq:sqwE}). The relevant eigenvectors of $U$ are $\ket{\psi(0)}$, $\ket{\text{d}_c}\ket{\beta_{(1,...,1)}}$, and $\ket{\alpha^{\vec k}_\pm}\ket{\beta_{\vec k}}$ (see~\cite{MR02,Por18book}), where $ \ket{\alpha_-^{\vec k}} = \ket{\alpha_+^{\vec k}}^*$,
\begin{eqnarray}
    \ket{\alpha_{+}^{\vec{k}}} = \frac{\textrm{e}^{\textrm{i}\theta}}{\sqrt{2}}\sum_{a=1}^n \left(\frac{k_a}{\sqrt{k}} - \textrm{i}\frac{1-k_a}{\sqrt{n-k}}\right) \ket{a},
    \label{pq_alpha_k+-}
\end{eqnarray}
for $0< k < n$, where $k$ is the Hamming weight of $\vec{k}$, $k_a$ is the $a$-th entry of $\vec{k}$, and $\cos\theta=\sqrt{k/n}$; and
\begin{equation}\label{pq_beta_vec_k_3}
  \ket{\beta_{\vec{k}}}\equiv\frac{1}{\sqrt{N}}\sum_{\vec{v}=0}^{2^n-1}(-1)^{\vec{k}\cdot \vec{v}}\ket{\vec{v}},
\end{equation}
where $\vec{k}\cdot \vec{v}= \sum_j \vec{k}_j \vec{v}_j\mod 2$. The corresponding eigenvalues are 1, $-1$, and $\e^{\pm\i \omega_k}$. 
A (non-orthogonal) eigenbasis of $U$ has been summarized in Table~\ref{pq_tab_eigenvalvec}, but we use only the eigenvectors that have a nonzero overlap with the marked elements.
Using 
\begin{equation}\label{pq_D_alpha_veck_2}
	\left\langle \text{d}_c \Big| \alpha_{+}^{\vec{k}} \right\rangle = \left\langle \text{d}_c \Big|\alpha_-^{\vec{k}} \right\rangle = \frac{1}{\sqrt{2}}
\end{equation}
for $0< k < n$, we obtain $A = B = D = 0$, 
\begin{align}
	C &= -\frac{4n}{N^2} \sum_{\substack{\vec k = 1\\ \vec{k}\cdot\vec{v}_{0}  \text{ odd}}}^{N - 1} \frac{1}{k}  ,\label{eq:sqwC_hypercube_2} \\
	E &= \frac{n^2}{2N^2} \sum_{\substack{\vec k, \Vec{k'} = 1\\   (\vec{k} \oplus \vec{k'})\cdot\vec{v}_{0} \text{ odd}    }}^{N-1} \frac{1}{kk'}
	. \label{eq:sqwE_hypercube_2}
\end{align}
The asymptotic behavior of variables $C$ and $E$ is obtained in Appendix~\ref{appendB}, and is given by
\begin{align}
    C &= -\frac{4}{N}+O\left(\frac{1}{n}\right), \\
    E &= 1 + O\left(\frac{1}{n}\right),
\end{align}
for any location of the second marked vertex $\vec{v}_{0}$.

Using Eq.~(\ref{eq:lambdaequation}) in the asymptotic limit, we obtain
\begin{align}
    \lambda = \frac{2}{\sqrt N}.
\end{align}
Note that $\lambda\ll \omega_{k}$ when $N\gg 1$ because $\omega_{k}=O(1/\sqrt{n})$.
Using Eqs.~(\ref{eq:sinlam}), (\ref{eq:sinlamphi}), and (\ref{eq:Llmmp}),
in the asymptotic limit we obtain,
\begin{equation}
    \Lambda^\lambda_{m_0m_0} = \Lambda_{m_1m_1}^\lambda =
    - \Lambda^\lambda_{m_0m_1} = - \frac{1}{\sqrt N} + O\paren{\frac{1}{n}}.
\end{equation}
Using Eq.~(\ref{eq:Lpmmp}), we obtain
\begin{align}
    \Lambda'_{m_0m_0} = \Lambda'_{m_1m_1} = \Lambda'_{m_0m_1} + 1 =
        2 + O\paren{\frac{1}{n}}.
\end{align}
Replacing those results in Eq.~(\ref{eq:alpha}), we obtain $\alpha = \sqrt{N/8}$ and
$\braket{0'}{\lambda} = \braket{1'}{\lambda} = 1/\sqrt{8}$ asymptotically.
Using Eq.~\eqref{psi_lambda}, we obtain
\begin{equation}
    \braket{\psi(0)}{\lambda}= \frac{\i}{\sqrt 2} + O\paren{\frac{1}{n}}.
\end{equation}
The equation above can be used to show that the terms $\epsilon_m$ in Eq.~(\ref{eq:sqw_p(t)}) tend to zero when $N\rightarrow\infty$.
Using Eqs.~\eqref{eq:p(t)} and~(\ref{eq:p_succ}), we have asymptotically 
\begin{equation}
    p(t)= \frac{1}{2} \sin^2\paren{\frac{2t}{\sqrt N}}.
\end{equation}
Taking the running time as the optimal $t$, that is,
\begin{equation}
    t_\text{opt} = \frac{\pi\sqrt{N}}{4},
\end{equation}
we obtain asymptotically
\begin{equation}
    p_\text{succ} = \frac{1}{2}.
\end{equation}

\section{Numerical methods}\label{sec:hypercubeM3}

In this section, we show how to enhance numerical methods that estimate the time complexity of quantum walk-based search algorithms with multiple marked vertices using a representative example. We consider a Grover walk on hypercubes, as described in Sec.~\ref{sec:Hypercube}, but now we take an arbitrary number of marked vertices, and run Python simulations on an ordinary laptop. By using Eq.~(\ref{det_C}), we determine $\lambda$ numerically for an increasing number of marked vertices. Fig.~\ref{fig:1} shows $t_\text{opt}=\pi/(2\lambda)$ as a function of $n$ for $|M|$ equal to 3, 9, and 21. The plot shows that $t_\text{opt}$ scales as $\sqrt N$ and, by rescaling $N$ into $N/|M|$, all lines merge into the dashed line $\pi\sqrt{N}/(2\sqrt 2)$. This shows that the optimal number of steps right before measurement is the one given by Eq.~(\ref{eq:conj}) below. These results do not depend on the locations of the marked vertices.

\begin{figure}[t]
\includegraphics[width=0.48\textwidth]{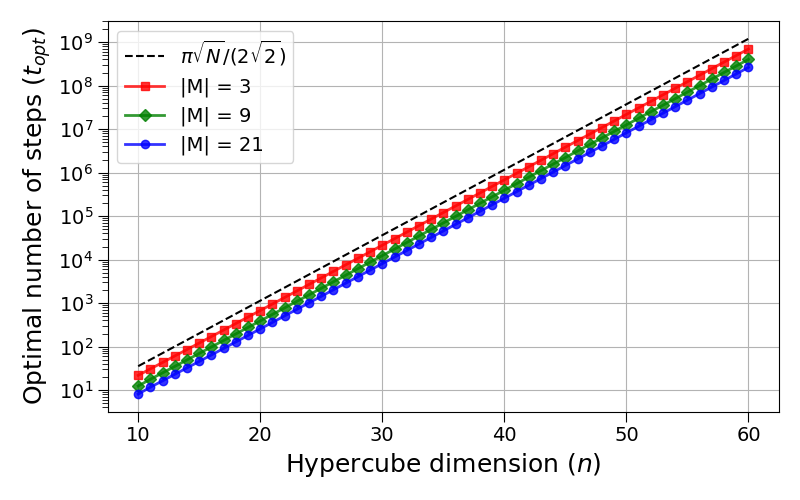}
\caption{Plot of the optimal number of steps $t_\text{opt}$ as a function of the hypercube dimension $n$ for an increasing number of marked vertices $|M|$, whose locations are chosen at random. }\label{fig:1}
\end{figure}

Next step is the analysis of the success probability $p_\text{succ}$ given by Eq.~(\ref{eq:p_succ}). Fig.~\ref{fig:2} shows $(0.5-p_\text{succ})$ as a function of $n$ for the same values of $|M|$ of Fig.~\ref{fig:1}, and $0.65/n^{1.056}$ as a function of $n$, which is a straight line in loglog scale obtained by curve fitting. We have eliminated values corresponding to $n<30$ because the high order terms of the asymptotic expansion of $p_\text{succ}$ play a relevant role for those values and cannot be fitted into a straight line. These numerical results show that the asymptotic success probability is 1/2. These results do not depend on the locations of the marked vertices.

\begin{figure}[t]
\includegraphics[width=0.48\textwidth]{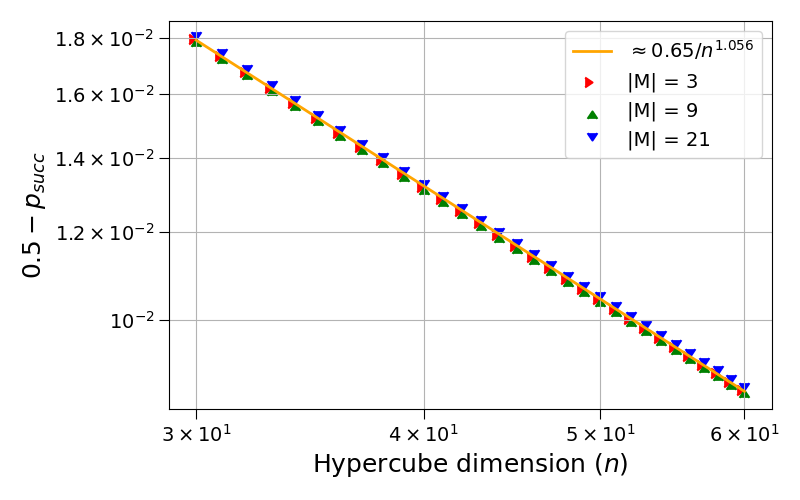}
\caption{Plot of $(0.5-p_\text{succ})$ as a function of the hypercube dimension $n$ for the same number of marked vertices of Fig.~\ref{fig:1}, where $p_\text{succ}$ is the success probability.
The axes are in loglog scale and the points are the tips of the triangles.}\label{fig:2}
\end{figure}

Without using Eq.~(\ref{det_C}), the only way to simulate the time evolution of the quantum walk requires the implementation of the coin and shift operators, which uses exponential resources as a function of $n$, and cannot be performed in the range of Fig.~\ref{fig:1} even in the largest supercomputers.
Our numerical results support the following conjecture:\vspace{5pt}

\noindent
\textbf{Conjecture.}  The asymptotic optimal running time for the Grover walk-based search algorithm on the hypercube with $|M|$ marked vertices is  \vspace{-5pt}
\begin{equation}\label{eq:conj}
    t_\text{opt}\,=\,\frac{\pi}{2\sqrt 2}\sqrt{\frac{N}{|M|}}
\end{equation}
and the asymptotic success probability is 1/2.\\

\section{Conclusions}\label{sec:conc}

We have developed an analytical method to calculate the time complexity of quantum walk-based search algorithms with multiple marked vertices. The method relies on two eigenvectors of the evolution operator associated with the eigenvalues that are closest to 1 but different from 1. The number of steps is given by $\pi/(2\lambda)$ in the simplest case when the evolution operator has real entries, where $\pm\lambda$ are the phases of the eigenvalues $\exp{\pm\i\lambda}$. Usually, the success probability decreases as a function of the number of vertices, the hypercube and the complete graph being notable exceptions. This method employs some hypotheses in order to proceed with the analytical calculations that must be checked on specific applications, and if the hypotheses are not confirmed, it means that the method cannot be used. In fact, our method may fail if the time complexity depends on all eigenvectors of the evolution operator. The method can also be used to speedup numerical analysis of search algorithms when the graph structure is too complex for an analytical approach.

We have applied our method to the Grover walk on the two-dimensional lattice and hypercube both with $N$ vertices, and we have shown that the optimal running time is $O(\sqrt{N\ln N})$ with success probability $O(1/\ln N)$ for the two-dimensional lattice and $O(\sqrt{N})$ with success probability $O(1)$ for the hypercube when they have two marked vertices. Since we have obtained the exact asymptotic limits for the running time and the success probability, we can draw further conclusions. When we compare our results with the corresponding ones for the two-dimensional lattice with one marked vertex, the behavior of the algorithm depends on the distance of the marked vertices. If we take them as far apart as possible, the success probability for the 2-marked case is the same as the 1-marked case but the running time is shorter by a factor of $\sqrt 2$. If we take the marked vertices as close as possible, the running time for the 2-marked case is the same as the 1-marked case but the success probability is smaller by a factor of $2$. Note that the presence of an extra marked vertex in the first case makes the searching easier and in the latter case makes it worse, different from what is usually expected.

When we compare our results with the corresponding ones for the hypercube with one marked vertex, the success probability is the same but the optimal number of steps for the 2-marked case is smaller by a factor of $\sqrt{2}$ for arbitrary locations of the marked vertices. Our numerical calculations show that those results can be extended for more marked vertices regardless of their locations. 

It would be interesting to apply our method to other quantum walk models, such as the staggered model~\cite{PSFG16}, and to analyze whether it can help to find exceptional configurations.

\appendix

\section{Asymptotic expressions for the two-dimensional lattice}\label{appendA}

In this Appendix we obtain simpler expressions for $C$ and $E$, and asymptotic expressions for $\lambda$.
Let us define the sums $S_1$ and $S_2$ as 
\begin{eqnarray}
	S_1 &=& \sum_{\substack{k,\ell=0\\ (k,\ell)\neq (0,0)}}^{\sqrt{N}-1}\frac{1}{1-\cos \theta_{k\ell} },
\end{eqnarray}
\begin{eqnarray}
	S_2 &=& \sum_{\substack{k,\ell=0\\ (k,\ell)\neq (0,0)}}^{\sqrt{N}-1}\frac{\sin^2\frac{\pi(kx_0+\ell y_0)}{\sqrt{N}}}{\sin^2 \frac{\pi k}{\sqrt{N}}+\sin^2 \frac{\pi \ell}{\sqrt{N}} }.
\end{eqnarray}
$S_1$ has the following bounds~\cite{Por18book,AKR05}
\begin{equation}\label{Append_sum}
     \frac{2}{\pi^2} N \ln N  \leq  S_1  \leq N \ln N
\end{equation}
up to $O(N)$ terms. 

Using Eq.~(\ref{eq:sqwC_latt2}) and the identity $\cos 2a=1-2\sin^2 a$, we obtain
\begin{equation}
    C \,=\, -\frac{8S_2}{N^2}. \label{eq:CS2}
\end{equation}
Using Eq.~(\ref{eq:sqwE_latt2}) and the antisymmetry of the sine function, the expression for $E$ can be simplified to
\begin{equation}\label{eq:E_S_2}
    E \,=\, \frac{4S_2}{N^2}\left(S_1-S_2\right). 
\end{equation}
Using Eq.~(\ref{eq:lambda}), we obtain
\begin{equation}
    \lambda \,=\, \frac{\sqrt 2}{\sqrt{S_1-S2}}.
\end{equation}
Let us proceed with two kinds of pairs of marked points.

\subsection*{Case $(x_0,y_0)=(1,0)$}

If $(x_0,y_0)=(1,0)$, then
\begin{eqnarray}
	S_2 &=& \sum_{\substack{k,\ell=0\\ (k,\ell)\neq (0,0)}}^{\sqrt{N}-1}\frac{\sin^2\frac{\pi k}{\sqrt{N}}}{\sin^2 \frac{\pi k}{\sqrt{N}}+\sin^2 \frac{\pi \ell}{\sqrt{N}} }.
\end{eqnarray}
Note that if we interchange $k$ and $\ell$ inside the sum, we obtain the same result, which can be used to obtain
\begin{equation}
    S_2=\frac{N-1}{2}.
\end{equation}
Using Eq.~(\ref{eq:S_1}), we obtain
\begin{equation}
    \lambda \,=\, \frac{\sqrt 2}{\sqrt{cN\ln N +O(N)}}.  
\end{equation}

\subsection*{Case $(x_0,y_0)=(\sqrt{N}/2,\sqrt{N}/2)$}

In this case, we assume that $\sqrt N$ is even. Replacing $(x_0,y_0)=(\sqrt{N}/2,\sqrt{N}/2)$ into $(S_1-S_2)$, we obtain
\begin{eqnarray}
	S_1-S_2 &=& \frac{1}{2}\sum_{\substack{k,\ell=0\\ (k,\ell)\neq (0,0)}}^{\sqrt{N}-1}\frac{1+(-1)^{k+l}}{\sin^2 \frac{\pi k}{\sqrt{N}}+\sin^2 \frac{\pi \ell}{\sqrt{N}} }.
\end{eqnarray}
Using $4k^2/\pi^2\le \sin^2 k\le k^2$, for $-\pi/2\le k\le \pi/2$ and $(k+\ell)^2/2\le k^2+\ell^2 \le (k+\ell)^2$, we obtain $S_1-S_2=cNS_3+O(N)$, where~\cite{note1}
\begin{equation}
    S_3 = \sum_{\substack{k,\ell=0\\ (k,\ell)\neq (0,0)}}^{\sqrt{N}/2}\frac{1+(-1)^{k+l}}{\left({k}+{\ell}\right)^2 }= \frac{\ln N}{2}+O\left(1\right).
\end{equation}Using Eq.~(\ref{eq:S_1}), we obtain
\begin{equation}
    \lambda \,=\, \frac{2}{\sqrt{cN\ln N +O(N)}}.  
\end{equation}

\section{Asymptotic expressions for the hypercube}\label{appendB}

In this Appendix we obtain asymptotic expressions for $C$ and $E$ given by Eqs.~(\ref{eq:sqwC_hypercube_2}) and~(\ref{eq:sqwE_hypercube_2}), respectively. Let us define $S_\text{odd}$ and $S_\text{even}$ as
\begin{eqnarray*}
    S_\text{odd} &=& \frac{n}{N}\sum_{\substack{\vec k = 1\\ \vec{k}\cdot\vec{v}_{0} \text{ odd}\substack}}^{N - 1} \frac{1}{\big|\vec k\big|},\\
    S_\text{even} &=& \frac{n}{N}\sum_{\substack{\vec k = 1\\ \vec{k}\cdot\vec{v}_{0} \text{ even}\substack}}^{N - 1} \frac{1}{\big|\vec k\big|},
\end{eqnarray*}
where $\big|\vec k\big|$ is the Hamming weight of $\vec k$ and $N=2^n$.
Using Eq.~(\ref{eq:sqwC_hypercube_2}), we have that 
$$
C=-\frac{4\,S_\text{odd}}{N}.
$$
Using Eq.~(\ref{eq:sqwE_hypercube_2}) and that $(\vec{k}\oplus\vec{k'})\cdot\vec{v}_{0}$ is odd in two cases: (1)~$\vec{k}\cdot\vec{v}_{0}$ is even and $\vec{k'}\cdot\vec{v}_{0}$ is odd, (2)~$\vec{k}\cdot\vec{v}_{0}$ is odd and $\vec{k'}\cdot\vec{v}_{0}$ is even, we obtain
\begin{equation*}
    	E =S_\text{odd}\,S_\text{even}.
\end{equation*}
Now we use the following lemmas to show that asymptotically $S_\text{odd}=\,S_\text{even}=1$ for any $\vec v_0\neq (0,...,0)$.

\begin{lemma}\label{appen_lem_2}
Let $n$ be a positive integer. Then
\begin{equation*}
   \frac{n}{2^n} \sum_{k=1}^{n} \frac{1}{k} \binom{n}{k} \,=\,
   2 \sum_{k=0}^\infty \frac{a_k}{n^k} -\frac{n}{2^n}\left(H_n+\frac{2}{n}\right)- O\left(\frac{1}{2^{n}}\right) ,
\end{equation*}
where $a_k$ is the $k$-th ordered Bell number (Fubini numbers~\cite{note2}, where $a_0=a_1=1,a_2=3,...$), and $H_n$ is the $n$-th harmonic number.
\end{lemma}

\begin{proof}
Define~\cite{note3}
\begin{equation*}
     s=\frac{n}{2^{n+1}} \sum_{k=1}^{n} \frac{1}{k} \binom{n}{k} + \frac{n}{2^{n+1}}\left(H_n+\frac{2}{n}\right).
\end{equation*}
Using the definition of $H_n$ and identity
\[
\sum_{k=1}^n\frac{1}{k}\binom{n}{k}=\sum_{k=1}^n\frac{2^k-1}{k},
\]
we obtain
\begin{equation*}
     s=\frac{n}{2^{n+1}} \left( \frac{2}{n}+\sum_{k=1}^{n} \frac{2^k}{k}\right) .
\end{equation*}
Reversing the sum and using $1/2^n=1-\sum_{i=0}^{n-1}1/2^{i+1}$, we obtain
\begin{equation*}
     s=1+\sum_{i=0}^{n-1} \frac{i}{2^{i+1}(n-i)}.
\end{equation*}
Using $1/(n-i)=\sum_{k=0}^\infty i^k/n^{k+1},$ asymptotically we obtain
\begin{equation*}
     s= \sum_{k=0}^\infty \frac{a_k}{n^k}- O\left(\frac{1}{2^{n+1}}\right),
\end{equation*}
where
\begin{equation*}
     a_k=\frac{1}{2}\sum_{i=0}^{\infty} \frac{i^k}{2^{i}},
\end{equation*}
is the $k$th Fubini number.
\end{proof}


\begin{lemma}\label{appen_lem_3}
Let $n$ be a positive integer, $\vec v$ be a $n$-bit vector, and $v$ be the Hamming weight of $\vec v$. Then
\begin{eqnarray*}
\sum_{\substack{\vec k = 1\substack}}^{2^n - 1} \frac{(-1)^{\vec{k}\cdot\vec{v}}}{\big|\vec k\big|}   &=&  {\binom{n}{v}}^{-1}\sum_{k=1}^{n-v}\frac{1}{k}\binom{n}{v+k} -\left( \psi \left( v+1 \right)+\gamma \right),
\end{eqnarray*}
where $|\vec k|$ is the Hamming weight of $\vec k$, $\psi$ is the digamma function, $\gamma$ is Euler's constant, and the left-hand sum runs over all $n$-bit vectors different from the null vector.
\end{lemma}

\begin{proof}
After permuting $k_1$,...,$k_n$, we obtain
\[
\sum_{\substack{\vec k = 1\substack}}^{2^n - 1} \frac{(-1)^{\vec{k}\cdot\vec{v}}}{\big|\vec k\big|} =
\sum_{k_1=0}^1\cdots\sum_{k_v=0}^1\cdots\sum_{k_n=0}^1 \frac{(-1)^{k_1+...+k_v}}{k_1+...+k_n}
\]
with an additional constraint that $(k_1,...,k_n)\neq (0,...,0)$. This result can be simplified to
\[
\sum_{\substack{\vec k = 1\substack}}^{2^n - 1} \frac{(-1)^{\vec{k}\cdot\vec{v}}}{\big|\vec k\big|} =
\sum_{\vec{k}_b=1}^{2^{n-v}-1}\sum_{\vec{k}_a=0}^{2^{v}-1}\frac{(-1)^{|\vec{k}_a|}}{|\vec{k}_a|+|\vec{k}_b|} + \sum_{\vec{k}=1}^{2^{v}-1}\frac{(-1)^{|\vec{k}|}}{|\vec{k}|},
\]
where the second term on the right-hand side is obtained from the case $(k_{v+1},...,k_n)=(0,...,0)$. The sums with vector indices can be simplified into
\[
\sum_{\substack{\vec k = 1\substack}}^{2^n - 1} \frac{(-1)^{\vec{k}\cdot\vec{v}}}{\big|\vec k\big|} =
\sum_{k'=1}^{n-v}\sum_{k=0}^{v}\frac{(-1)^k}{k+k'}\binom{v}{k}\binom{n-v}{k'}+\sum_{k=1}^v\frac{(-1)^k}{k}\binom{v}{k}.
\]
After simplifying the sums and using some binomial identities, we complete the proof.
\end{proof}

Using 
$$ \sum_{\substack{\vec k = 1\substack}}^{N - 1} \frac{1}{\big|\vec k\big|}=\sum_{\substack{k=1}}^{n} \frac{1}{k} \binom{n}{k},$$
we obtain
\begin{eqnarray*}
   S_\text{even}+S_\text{odd}&=& \frac{n}{N}\sum_{\substack{k=1}}^{n} \frac{1}{k} \binom{n}{k}, \\
   S_\text{even}-S_\text{odd}&=& \frac{n}{N}\sum_{\substack{\vec k = 1\substack}}^{N - 1} \frac{(-1)^{\vec{k}\cdot\vec{v}_{0}}}{\big|\vec k\big|}.
\end{eqnarray*}
Using Lemma~\ref{appen_lem_2}, we conclude that 
$$
 S_\text{even}+S_\text{odd} = 2 + O\left(\frac{1}{n}\right).
$$
To find the asymptotic expression of $ S_\text{even}-S_\text{odd}$, we split the analysis in the following cases.

\subsubsection*{Case 1: $v$ fixed}

If the Hamming weight $v$ of $\vec v_0$ is fixed, that is, $v$ does not depend on $n$, we use the inequality
\begin{equation*}
    \sum_{k=1}^{n-v}\frac{1}{k}\binom{n}{v+k} \le \sum_{k=1}^{n-v}\binom{n}{v+k}=  2^n-\sum_{k=1}^v\binom{n}{k}-1
\end{equation*}
and Lemma~\ref{appen_lem_3} to show that
$$
 S_\text{even}-S_\text{odd} = O\left(\frac{1}{n}\right),
$$
if $v> 1$. We obtain the same result for the case $v=1$ by using the inequality
\begin{align*}
    \sum_{k=1}^{n-1}\frac{1}{k}\binom{n}{k+1} &\le \frac{3}{n+1}\sum_{k=1}^{n-1}\binom{n+1}{k+2}\\
    &=  \frac{3}{n+1}\left(2\,2^n-\frac{n(n+1)}{2}-2\right).
\end{align*}

\subsubsection*{Case 2: $v$ depends on $n$}

Here we assume that $v$ depends on $n$ and $\lim_{n\rightarrow\infty}v=\infty,$ where $v=|\vec v_0|$.
Using the inequality\vspace{-5pt}
\[
\binom{n}{v+k}\le \binom{n}{v}\binom{n-v}{k},
\]
we have
\begin{equation*}
     \binom{n}{v}^{-1}\sum_{k=1}^{n-v}\frac{1}{k}\binom{n}{v+k} \le  \sum_{k=1}^{n-v}\frac{1}{k}\binom{n-v}{k}.
\end{equation*}
Using the inequality above, Lemma~\ref{appen_lem_2} (substituting $n-v$ for $n$), and  Lemma~\ref{appen_lem_3}, we obtain
$$
 S_\text{even}-S_\text{odd} = O\left(\frac{1}{n}\right),
$$
when $1\le v< n$. We obtain the same result for the case $v=n$ by using  Lemma~\ref{appen_lem_3} alone.

\onecolumngrid

\

\section{Case $|M|=3$}

Suppose that $M=\{m_{0},m_{1},m_{2}\}$. Eq.~(\ref{det_C}) implies that 
\begin{equation}
\sum_{k k' k''}b^\lambda_k b^\lambda_{k'} b^\lambda_{k''}c_{k k'}=0,
\end{equation}
where
\begin{equation}
c_{k k' k''}\,=\,\psi_k(m_{0})\psi_{k'}(m_{1})\psi_{k''}(m_{2})\det\left(\psi^*_k(m)\right).
\end{equation}
Suppose that $\lambda\ll \phi_\textrm{min}$ when $N\gg 1$, where $\phi_\textrm{min}$ is the smallest positive value of $\phi_{k}$. We will check the validity of those assumptions in specific applications. Let us split the sum into the following parts
\begin{equation}\label{exactsum_kl_3}
 \sum_{k k'}=\sum_{\substack{\phi_k=0\\ \phi_{k'}=0\\\phi_{k''}=0}}+\sum_{\substack{\phi_k=0\\ \phi_{k'}= 0\\ \phi_{k''}\neq 0}}+\sum_{\substack{\phi_k=0\\ \phi_{k'}\neq 0\\ \phi_{k''}= 0}}+\sum_{\substack{\phi_k\neq 0\\ \phi_{k'}= 0\\ \phi_{k''}= 0}}+\sum_{\substack{\phi_k\neq 0\\ \phi_{k'}\neq 0\\\phi_{k''}=0}}+\sum_{\substack{\phi_k\neq 0\\ \phi_{k'}= 0\\ \phi_{k''}\neq 0}}+\sum_{\substack{\phi_k=0\\ \phi_{k'}\neq 0\\ \phi_{k''}\neq 0}}+\sum_{\substack{\phi_k\neq 0\\ \phi_{k'}\neq 0\\ \phi_{k''}\neq 0}},
\end{equation} 
corresponding to the sum of terms such that $\phi_k=0$ or $\phi_k\neq 0$. Using Eqs.~(\ref{eq:sinlam}) and~(\ref{eq:sinlamphi}),
we obtain 
\begin{equation}\label{eq:lambdaequation3}
	\frac{A}{\lambda^3}+\frac{B}{\lambda^2}+\frac{C}{\lambda}+D+E\lambda= O(\lambda^2),
\end{equation}
where
\begin{eqnarray}
	A &=&  8\sum_{\substack{\phi_k=0\\ \phi_{k'}=0 \\ \phi_{k''}=0}} c_{k k' k''}  ,\label{eq:sqwA3} 
\end{eqnarray}
\begin{eqnarray}
	B &=& 4\sum_{\substack{\phi_k=0\\ \phi_{k'}= 0\\ \phi_{k''}\neq 0}} a_{k''} c_{k k' k''}\sin\phi_{k''} + 4
	\sum_{\substack{\phi_k=0\\ \phi_{k'}\neq 0\\ \phi_{k''}= 0}} a_{k'} c_{k k' k''}\sin\phi_{k'}+ 4
	\sum_{\substack{\phi_k\neq 0\\ \phi_{k'}= 0\\ \phi_{k''}= 0}} a_{k} c_{k k' k''}\sin\phi_{k},\label{eq:sqwB3} 
\end{eqnarray}
\begin{eqnarray}
C &=& 2\sum_{\substack{\phi_k\neq 0\\ \phi_{k'}\neq 0\\\phi_{k''}=0}} a_{k} a_{k'} c_{k k' k''} \sin\phi_k \sin\phi_{k'}
	+2\sum_{\substack{\phi_k\neq 0\\ \phi_{k'}= 0\\ \phi_{k''}\neq 0}} a_{k} a_{k''} c_{k k' k''} \sin\phi_k \sin\phi_{k''}+\nonumber \\
	&&2\sum_{\substack{\phi_k=0\\ \phi_{k'}\neq 0\\ \phi_{k''}\neq 0}} a_{k'} a_{k''} c_{k k' k''} \sin\phi_{k'} \sin\phi_{k''}+
	4\sum_{\substack{\phi_k=0\\ \phi_{k'}= 0\\ \phi_{k''}\neq 0}} a_{k''} c_{k k' k''} +4\sum_{\substack{\phi_k=0\\ \phi_{k'}\neq 0\\ \phi_{k''}= 0}} a_{k'} c_{k k' k''}+ \nonumber \\
	&&	4\sum_{\substack{\phi_k\neq 0\\ \phi_{k'}= 0\\ \phi_{k''}= 0}} a_{k} c_{k k' k''},\label{eq:sqwC3} 
\end{eqnarray}
\begin{eqnarray}
D &=& 2\sum_{\substack{\phi_k\neq 0\\ \phi_{k'}\neq 0\\\phi_{k''}=0}} a_{k} a_{k'} c_{k k' k''} (\sin\phi_k + \sin\phi_{k'})
	+2\sum_{\substack{\phi_k\neq 0\\ \phi_{k'}= 0\\ \phi_{k''}\neq 0}} a_{k} a_{k''} c_{k k' k''} (\sin\phi_k + \sin\phi_{k''})+\nonumber \\
	&&2\sum_{\substack{\phi_k=0\\ \phi_{k'}\neq 0\\ \phi_{k''}\neq 0}} a_{k'} a_{k''} c_{k k' k''} (\sin\phi_{k'} + \sin\phi_{k''})+
	\sum_{\substack{\phi_k\neq 0\\ \phi_{k'}\neq 0\\ \phi_{k''}\neq 0}}\,a_{k}\,a_{k'}\,a_{k''}c_{k k' k''}\sin\phi_k\sin\phi_{k'}\sin\phi_{k''},\label{eq:sqwD3} 
\end{eqnarray}
\begin{eqnarray}
E &=& \sum_{\substack{\phi_k\neq 0\\ \phi_{k'}\neq 0\\ \phi_{k''}\neq 0}} a_{k}a_{k'}a_{k''} c_{k k' k''}\left( \sin\phi_k  \sin\phi_{k'} + \sin\phi_k  \sin\phi_{k''} + \sin\phi_{k'}  \sin\phi_{k''}  \right)   + \nonumber \\ 
	&&2\sum_{\substack{\phi_k\neq 0\\ \phi_{k'}\neq 0\\\phi_{k''}=0}} a_{k} a_{k'} c_{k k' k''} + 2\sum_{\substack{\phi_k\neq 0\\ \phi_{k'}= 0\\ \phi_{k''}\neq 0}} a_{k} a_{k''} c_{k k' k''} + 2\sum_{\substack{\phi_k=0\\ \phi_{k'}\neq 0\\ \phi_{k''}\neq 0}} a_{k'} a_{k''} c_{k k' k''},\label{eq:sqwE3}
\end{eqnarray}
and $a_k$ is given by Eq.~(\ref{eq:a_k}).

\section*{Acknowledgements}
 This study was financed in part by CAPES, FAPERJ grant number CNE E-26/202.872/2018, and CNPq grant numbers 407635/2018-1 and 140897/2020-8.

\end{document}